\documentclass[final,3p,times]{elsarticle}
\usepackage{amsmath,amssymb,amsfonts}
\usepackage[mathscr]{eucal}
\usepackage{amsthm}
\usepackage[arrow]{xy}
\usepackage{stmaryrd}

\usepackage{color}

\newtheorem{proposition}{Proposition}

\newtheorem{corollary}{Corollary}
\newtheorem{lemma}{Lemma}

\theoremstyle{definition}
\newtheorem{remark}{Remark}

\newtheorem{definition}{Definition}

\newcommand{\ff}{\mathfrak}
\newcommand{\C}{\mathcal}
\newcommand{\p}{\partial}

\newcommand{\F}[2]{\frac{#1}{#2}}

\renewcommand{\ker}{\mathrm{ker}}
\newcommand{\im}{\mathrm{im}}

\newcommand{\Hom}{\mathrm{Hom}}

\newcommand{\fa}{\mathfrak{A}}

\newcommand{\gh}{\mathrm{gh}}
\newcommand{\Deg}{\mathrm{Deg}}
\newcommand{\res}{\mathrm{res}}
\newcommand{\rank}{\mathrm{rank}}

\journal{Geometry and Physics}

\begin{document}

\begin{frontmatter}



\title{Lifting a Weak Poisson Bracket to the Algebra of Forms}


\author[label1]{S.~Lyakhovich}
\ead{sll@phys.tsu.ru}
\author[label2]{M.~Peddie}
\ead{matthew.peddie@manchester.ac.uk}
\author[label1]{A.~Sharapov}
\ead{sharapov@tsu.ru}

\address[label1]{Department of Quantum Field Theory, Tomsk State University, Tomsk 634050, Russia.}
\address[label2]{School of Mathematics,  University of Manchester, Oxford Road,  Manchester   M13 9PL,  UK.}

\begin{abstract}
We detail the construction of a weak Poisson bracket over a submanifold $\Sigma$ of a smooth manifold $M$ with respect to a local foliation of this submanifold. Such a bracket satisfies a weak type Jacobi identity but may be viewed as a usual Poisson bracket on the space of leaves of the foliation. We then lift this weak Poisson bracket to a weak odd Poisson bracket on the odd tangent bundle $\Pi TM$, interpreted as a weak Koszul bracket on differential forms on $M$. This lift is achieved by encoding the weak Poisson structure into a homotopy Poisson structure on an extended manifold, and lifting the Hamiltonian function that generates this structure. Such a construction has direct physical interpretation. For a generic gauge system, the submanifold $\Sigma$ may be viewed as a stationary surface or a constraint surface, with the foliation given by the foliation of the gauge orbits. Through this interpretation, the lift of the weak Poisson structure is simply a lift of the action generating the corresponding BRST operator of the system.
\end{abstract}


\end{frontmatter}


\section{Introduction and Background}
When considering classical gauge systems, one usually starts with a Lagrangian or Hamiltonian function which may then be used to derive the equations of motion through the least action principle. In most cases this function is known, but in those that are not, it is known that the existence of a classical BRST differential allows one to identify the equations of motion and the gauge symmetries independently of the Lagrangian. Such a BRST differential is a homological vector field on an appropriately extended manifold which encodes the gauge system \cite{HT92}. In \cite{LS05}, a general geometric set-up for an arbitrary gauge system was described, and an embedding of such a system into an extended manifold was constructed. This is without reference to any Lagrangian or Hamiltonian function and merely assumes the existence of the equations of motion, together with the gauge symmetries present. This geometric construction introduced a weak Poisson bracket, a Poisson bracket that satisfies a weak Jacobi identity, which allows a quantisation of these gauge systems which may not be Lagrangian or Hamiltonian. This paper is about lifting this weak Poisson bracket to the algebra of differential forms by utilising the described embedding into the extended manifold.

For an arbitrary smooth manifold $M$ with local coordinates $(x^i)$, let the system of equations $T^a(x)=0$, $a=1,\ldots,k$, define a smooth submanifold $\Sigma$ of codimension $k$. Choose $n$ linearly independent vector fields $R_\alpha$ on $M$ such that they are tangent to $\Sigma$. (In fact, linear independence is not necessary but simplifies the exposition, see remark \ref{Remark reducibility in gt}.) For a corresponding gauge system in the Lagrangian formalism, the equations $T^a(x)=0$ may be identified with the equations of motion, possibly derived from an action principle, whilst the manifold $M$ is understood as the space of trajectories in a configuration space of the system. The vector fields $R_\alpha$ then generate the gauge symmetries of the action. If we consider the  Hamiltonian formalism, then the surface $\Sigma$ may be identified with the constraint surface given by the constraint equations $T^a(x)=0$. The vector fields correspond to the gauge generators which define a foliation of the submanifold $\Sigma$ into gauge orbits identified with the integral submanifolds.

To obtain this foliation, it is required for the vector fields to form an integrable distribution over $\Sigma$; they must satisfy the commutation relation
    \begin{equation}\label{vector field relation}
    [R_\alpha,R_\beta] = f^\gamma_{\alpha\beta} R_\gamma + T^a X_{a\alpha\beta},
    \end{equation}
for smooth functions $f^\gamma_{\alpha\beta}$ and vector fields $X_{a\alpha\beta}$ on $M$. Notice that the presence of the constraint terms $T^a$ means that this is an open Lie algebra over $M$ that closes only over $\Sigma$. The space of leaves $N$ of the foliation of $\Sigma$ generated by this integrable distribution gives the true physical degrees of freedom when viewed as a gauge system. Because of this, physically interesting objects are those that descend to the leaf space, i.e. those which are constant over the gauge orbits which are identified with the integral submanifolds. In the work \cite{LS05}, a function or multivector field was defined to be projectible precisely when it may be considered as a smooth contravariant tensor field on the leaf space $N$. In general, the space $N$ may not be smooth. When we say smooth in this sense, we refer to those smooth tensor fields $\C{T}$ on $M$ which are constant along the integral submanifolds:
    \[\C{L}_{R_\alpha}\C{T} = \C{T}^\alpha R_\alpha + T^a\C{T}_a,\]
for smooth tensor fields $\C{T}_a, \C{T}^\alpha$; here $\C{L}_{R_\alpha}\C{T}$ is the Lie derivative along the vector field $R_\alpha$. Therefore, we should consider projectible multivector fields as physically interesting. Indeed, examples include: vector fields $R_\alpha$ that correspond to the gauge generators, vector fields that introduce dynamics into the gauge system, and bivector fields which can induce Poisson structures in the algebra of functions of $N$. (In physical literature the algebra of functions on $N$, identified with projectible functions on $M$, is called the algebra of physical observables.) Such a Poisson structure on $N$, specified by an appropriate projectible bivector field on $M$, induces a weak Poisson bracket on $M$. This is a Poisson bracket on $M$ corresponding to the same projectible bivector field, but which satisfies only a weak Jacobi identity; a Jacobi identity that holds over $\Sigma$ up to terms proportional to the vector fields $R_\alpha$.

An embedding of this foliation together with a weak Poisson bracket was detailed, \cite{LS05}, into which the information was encoded into a single function $S$ on an extended manifold. This embedding corresponds to reformulating the theory in terms of the BRST language \cite{HT92}. The function $S$ is called the master function and corresponds to the action associated to the gauge system. Such a function contains all the information present in the gauge system, and generates the BRST operator, a homological vector field on the extended manifold. This extended manifold is the original manifold $M$ appropriately extended by ghost variables. The notion of projectibility was then entwined with the BRST operator of the theory, where it was shown that projectible multivector fields are cocycles of this BRST operator considered as a differential on the larger algebra of functions.

The BRST operator may be considered as a unary bracket in a homotopy Poisson structure (a $P_\infty$-structure), which is generated by the master function $S$ on the extended manifold. The weak Poisson bracket on $M$ embeds into this homotopy Poisson structure, and is recovered as the restriction of the binary bracket to $M$. In \cite{KV08} an explicit construction was given to canonically lift a homotopy Poisson structure to an odd homotopy Poisson structure (an $S_\infty$-structure) on the odd tangent bundle. By applying this construction, the even homotopy Poisson structure may be lifted to the odd tangent bundle of the extended manifold, to produce the corresponding odd homotopy structure. Through this we may define a weak Koszul bracket on the odd tangent bundle $\Pi TM$. This weak bracket is the restriction of the binary bracket in this odd homotopy Poisson structure, and corresponds to the weak Poisson bracket on $M$ in precisely the same way that the well-known Koszul bracket of forms corresponds to a Poisson structure.

The Koszul bracket is a natural odd extension of the usual even Poisson bracket. A natural even extension to the entire algebra of forms does not exist; however, in the works \cite{CV92, Mic85}, it was shown that a Poisson bracket induces a genuine even Poisson bracket in the space of co-exact forms - differential forms modulo the exact forms. The exterior differential $d$
 from co-exact forms into differential forms was shown to be a homomorphism of Lie algebras, taking the even Poisson bracket on co-exact forms to the odd Koszul bracket on exact forms. The exact forms in this case form an ideal in the graded Lie algebra of forms endowed with the Koszul bracket.

The extension of the algebra of physical
observables by differential forms now implies a proper
generalisation of the notion of a physical state. Regarding the
usual classical states, the points of a phase space, as
$0$-cycles in the sense of algebraic topology, it is natural to
consider the cycles of higher degrees represented by
higher-dimensional closed surfaces, possibly with singularities.
These may be viewed as sort of mixed states in classical mechanics,
when only part of the physical data is exactly known about the system.
Integration of co-exact forms over cycles then
yields the measured values of  physical observables in such
mixed states. From this perspective, the construction of the weak Koszul
bracket on forms being proposed in this paper, extends the  previous results of \cite{CV92, Mic85}
to the case of constrained and gauge invariant dynamical systems.
More precisely, it may be shown that each weak Koszul bracket
induces a genuine Poisson bracket in the space of projectible
co-exact forms, and the dynamics of form-valued physical observables
are governed by a projectible vector field compatible with the
Poisson bracket. We will detail this construction
elsewhere.

In sections 2 and 3 we review the construction introduced in the work \cite{LS05}. We recall the weak Poisson bracket and the embedding of the geometrical set-up into the extended manifold. After, we discuss the relationship between projectible multivector fields and the homological vector field, the BRST operator of an associated gauge system. In section 4 the construction given in the work \cite{KV08} is applied to our setting. We lift the even homotopy Poisson structure generated by the master function $S$ to the odd tangent bundle of this extended manifold, and then define the weak Koszul bracket corresponding to the lift of the weak Poisson bracket. A vector field on the leaf space $N$ is introduced as a projectible vector field on $M$, which can be seen to provide dynamics to an associated gauge system. Using this, we discuss what it means for projectible differential forms to be constant over the flow generated by this vector field. It is shown that the weak Koszul bracket of two such projectible forms produces another, analogous to the Poisson bracket of two integrals of motion. In section 5 we provide some examples of weak Poisson brackets.

In this paper the language of supermanifolds will be used; we will denote the Grassmann parity of a homogeneous object by $\epsilon(\cdot)$ when we wish to be explicit. All the constructions naturally generalise to the case when the manifold $M$ is a supermanifold, however for notational convenience we will assume that $M$ is a usual manifold (bosonic). Further, the construction extends to the case when $M$ is strictly a supermanifold and the weak even Poisson bracket is now a weak odd Poisson bracket.

\section{Projectible Multivectors and a Weak Poisson Bracket}
\subsection{Constraints and Regularity}
The construction outlined in the introduction may be generalised and rephrased in the language of vector bundles. Let $E\rightarrow M$ be a vector bundle and fix a linear connection $\nabla_E$ in $E$. The submanifold $\Sigma$ may be identified with the zero locus of a section $T\in C^\infty(M,E)$,
    \[T = T^a(x)e_a,\]
where $\{e_a\}$ is a local frame for $E$ over $U\subset M$. In order for the zero locus of $T$ to define the smooth submanifold $\Sigma$, it is required that the map $\nabla_ET:TM\rightarrow E$, defined by the covariant derivative, must be of constant rank in a tubular neighbourhood $U_\Sigma$ of $\Sigma$. Locally $\nabla_ET$ is given by
    \[\nabla_ET = dx^i\nabla_{E,i}T^ae_a = dx^i\left(\p_iT^a + T^bA^a_{bi}\right)e_a,\]
where $A^a_{bi}$ are the connection coefficients. It is clear then that this map has constant rank in a tubular neighbourhood if and only if the matrix of partial derivatives $||\p_iT^a||$ does. If the functions $T^a$ are assumed to be linearly independent, then the section $T$ intersects the base $M$ transversally and the map defined must be of constant rank over $\Sigma$ only. With $T$ fixed, the covariant derivative of $T$ defines a bundle homomorphism
    \[\nabla_ET:TM\rightarrow E.\]
The vector fields $R_\alpha$ may also be encoded into a homomorphism from an appropriate vector bundle $F\rightarrow M$ into the tangent bundle $TM$. Notice from the identification $\Hom(F,TM)\cong F^*\otimes TM$, that this homomorphism of vector bundles defines and is defined by a section
    \[R = R^i_\alpha f^\alpha\otimes\F{\p}{\p x^i}\in C^\infty(M,F^*\otimes TM)\]
for some local frame $\{f^\alpha\}$ of $F^*$. The homomorphism $R$ is also required to have constant rank in some tubular neighbourhood of $\Sigma$. These constant rank requirements on $R$ and $\nabla_E T$ are called the \emph{regularity conditions} \cite{LS05, KLS05}. When restricted to $\Sigma$ the homomorphisms $R$ and $\nabla_ET$ define the exact sequence of vector bundles
    \begin{equation}\label{Exact sequence regularity}
    0\rightarrow F\xrightarrow{R} TM\xrightarrow{\nabla_ET} E\rightarrow 0.
    \end{equation}
In general, this sequence will not form a complex off $\Sigma$. The assumption that the vector fields $R_\alpha$ are linearly independent is equivalent to the homomorphism $R$ being injective. The image of $R$ in $TM$ forms a distribution which is integrable only over $\Sigma$.

\begin{remark}\label{Remark reducibility in gt}
The vector fields $R_\alpha$ may not be linearly independent. Indeed, we may choose $n$ linearly dependent vector fields, locally defined only over $\Sigma$. These vector fields identify with a generating set of a gauge algebra, which are not, in general, global independent generators. When identified with a generating set of a gauge algebra, this linear dependence is equivalent to the presence of reducibility in the gauge algebra. Likewise, the equations $T^a(x)=0$ do not need to be assumed to be independent equations. Relations may exist between these which are called the Noether identities in the physical literature \cite{HT92}. To compensate for this dependence, additional vector bundles need to be introduced on either side of $F$ and $E$ in \eqref{Exact sequence regularity} to incorporate these relations, and this chain of vector bundles is required to form an exact sequence over $\Sigma$, see \cite{KLS05}. A generating set for the gauge algebra is called complete if it contains all the information for the Noether identities \cite{HT92}. That the vector fields $R_\alpha$ can be identified with a complete set of gauge generators is equivalent to exactness in the middle term $TM$ of the sequence \eqref{Exact sequence regularity}.
\end{remark}

\subsection{Projectible Multivector Fields}
Recall that the foliation determined by the vector fields $R_\alpha$ has the associated space of leaves $N$. In general $N$ may not be a smooth manifold, however one can describe the ``smooth'' functions and tensor fields on $N$ in terms of $M$. To do this, we shall introduce the supermanifold $\Pi T^*M$, the odd cotangent bundle to $M$.

Let $\Pi T^*M$ have local coordinates $(x^i,x^*_i)$ where $\epsilon(x^*_i) = \epsilon(x^i) + 1 = 1$, and the $x^*_i$ transform as $x^*_i = \F{\p x^{i'}}{\p x^i} x^*_{i'}$ for a change of coordinates $x = x(x')$. These coordinates are the odd momenta and the subalgebra $\fa(M)\subset C^\infty(\Pi T^*M)$ of fibrewise polynomial functions carries a natural grading by the odd momentum degree. Such fibrewise polynomial functions are (in the case of usual manifolds (bosonic) $M$,) naturally identified with multivector fields on $M$ by the odd isomorphism $\p_a\mapsto x^*_a$. This is a homomorphism of Lie algebras taking the Schouten bracket of multivector fields $\llbracket-,-\rrbracket$ to the canonical non-degenerate Poisson bracket on $\Pi T^*M$ induced from the canonical odd symplectic structure. These will be freely identified, for example, we will interchange the non-degenerate Poisson bracket and the Schouten bracket without reference.

Define an ideal $\mathfrak{I}$ in $\fa(M)$ generated by the functions $T^a$ and the vector fields $R_\alpha$,
    \[\mathfrak{I} = \langle T^a,R_\alpha\rangle.\]
The ideal $\ff{I}$ is closed under the Schouten bracket by equation \eqref{vector field relation} and since the $R_\alpha$ are tangent to $\Sigma$. In terms of the sections $T$ and $R$, it is defined by the images of the associated maps $\hat{T}$ and $R$, where $R$ is the homomorphism $R:F\rightarrow TM$ as before, and $\hat{T}:C^\infty(M,E^*)\rightarrow C^\infty(M)$ is obtained by fixing $T$ (so giving all linear combinations of the functions $T^a$ which uses the natural pairing of sections $C^\infty(M,E)\times C^\infty(M,E^*)\rightarrow C^\infty(M)$).
Multivector fields which are elements of this ideal will be called \emph{trivial} multivectors. Trivial multivectors are proportional to linear combinations of the functions $T^a$ and the vector fields $R_\alpha$, and so when restricted to the submanifold $\Sigma$ will be proportional only to the vector fields $R_\alpha$. An equivalence relation $\sim$ may be defined in the algebra $\fa(M)$; two multivectors $U$ and $V$ are said to be equivalent, $U\sim V$, if their difference lies in $\ff{I}$,
    \begin{equation}\label{equivalence of multivectors}U - V \in \mathfrak{I}.\end{equation}
A multivector field $U$ will be called \emph{projectible} if
    \begin{equation}\label{projectibility relations}\llbracket U,\ff{I}\rrbracket \subset \ff{I}.\end{equation}

In particular, projectible multivector fields are those that are tangent to $\Sigma$ and are constant over the integral submanifolds of the foliation defined by the vector fields $R_\alpha$. They form a closed subalgebra $\fa_P(M)\subset\fa(M)$ and define a Poisson normaliser of the ideal $\ff{I}\subset\fa_P(M)$. The algebra of projectible multivector fields on $M$ inherits the momentum grading from the larger algebra of multivector fields and can be written $\fa_P(M) = \oplus_{k\geq0}\fa^k_P(M)$. Define the smooth multivector fields on the leaf space $N$ then as the quotient
    \begin{equation*}
    \fa(N) = \fa_P(M)/\ff{I},
    \end{equation*}
of the projectible multivector fields modulo the trivial ones. These are precisely the multivector fields on $\Sigma$ which are constant over the integral submanifolds of the foliation under the equivalence relation defined above.

Some remarkable subspaces of $\fa(N)$ are $\fa^0(N)$ and $\fa^1(N)$. The first can be identified with the set of smooth functions on $N$, $\fa^0(N)\cong C^\infty(N)$ and is referred to as \emph{the algebra of physical observables} in physics. The second space is the space of vector fields on $N$. Such vector fields determine $1$-parameter groups of automorphisms of the algebra $C^\infty(N)$. For a projectible vector field $V\in\fa^1_P(M)$ and a function $F\in C^\infty(N)$, define the Lie derivative of $F$ along $V$ by
      \begin{equation}\label{time evolution}
      \dot{F} = \llbracket V,F\rrbracket.
      \end{equation}
The function $\dot F$ will remain projectible for as long as $V$ is projectible. In other words, $\fa^0(N)\cong C^\infty(N)$ is a module over the Lie algebra $\fa^1(N)$ and each projectible vector field defines a one-parameter group of automorphisms of $C^\infty(N)$.

\begin{remark}\label{dynamics remark}
When identified with a gauge system, the leaf space $N$ is the reduced phase space, that is, the phase space modulo the gauge equivalence. The smooth functions in $C^\infty(N)$ are then the measurable observables justifying the name of the algebra of physical observables. They are those functions which are invariant under gauge transformation. Any projectible vector field $V$ as above then provides physical dynamics to the space, and the $1$-parameter group of automorphisms \eqref{time evolution} determines the time evolution of the observable $F$ under the dynamics of $V$.
\end{remark}

It is of interest to quantise the algebra $C^\infty(N)$ consistently, and in order to do this a construction was given in \cite{LS05} to equip this algebra with a Poisson bracket. A natural way to do this is to introduce a projectible bivector field $P\in\fa^2_P(M)$ which satisfies a weak type Jacobi identity
    \begin{equation}\label{weak Jacobi}\llbracket P,P\rrbracket\in \ff{I}.\end{equation}
Such a bivector field induces a derived \emph{weak Poisson bracket} on the algebra of functions $C^\infty(M)$ by the following formula:
    \[\{F,G\} := \llbracket\llbracket P,F\rrbracket,G\rrbracket.\]
It is called a weak Poisson bracket since the Jacobi identity holds only up to trivial multivector fields, equivalent to the weak commutation relation \eqref{weak Jacobi} of the Poisson bivector $P$. The algebra $C^\infty(N)$ however receives a Poisson bracket in the usual sense since we pass to the quotient by the trivial multivector fields. A projectible vector field $V\in\fa^1_P(M)$ will be called \emph{weakly Poisson} if it preserves the weak Poisson structure in the sense that $\llbracket V,P\rrbracket\in\ff{I}$. A weak Poisson vector field $V$ together with the Poisson bivector $P$ defines a Hamiltonian structure on the space $N$. The algebra of functions $C^\infty(N)$ is then a Poisson algebra together with a $1$-parameter group of Poisson automorphisms given by \eqref{time evolution}.

\section{The Extended Manifold and the Master Function}
\subsection{The Extended Manifold}
The construction of the previous section may be encoded into a single function on an extended manifold. This function generates a first order differential operator called a homological vector field, which is precisely the BRST operator \cite{HT92} of an associated gauge system. This BRST operator encodes the information contained in the gauge system, and allows one to study the system through the cohomology of the operator.

To begin, the original manifold $M$ must be extended to include additional variables. Define the extended manifold $\C{M}$ to be the total space of the vector bundle $\Pi E\oplus \Pi F\rightarrow M$, which will complement the original variables $(x^i)$ with new odd coordinates $(\eta^a,c^\alpha)$ from the fibres of the bundles $\Pi E$ and $\Pi F$ respectively. The supermanifold $\C{M}$ must then be further extended to the odd cotangent bundle $\C{N}=\Pi T^*\C{M}$ to include the momenta variables with reversed parity: $(X^*_i,\eta^*_a,c^*_\alpha)$. We would like to interpret polynomial functions on $\C{N}$ as multivector fields on the extended manifold $\C{M}$, however, notice that under a change of coordinates $x = x(x')$, the odd momenta $X^*_i$ transform in the following way:
    \begin{equation*}
    X^*_i = J^{i'}_iX^*_{i'} + (-1)^{iB}z^{B'}T^B_{B'}\p_i\left(T^{A'}_B\right)z^*_{A'},
    \end{equation*}
where we write
    \[z^A = (\eta^a,c^\alpha), \quad z^*_A = (\eta^*_a,c^*_\alpha), \quad J^{i'}_i(x) = \F{\p x^{i'}}{\p x^i}(x), \quad\mbox{and}\quad z^A = z^{A'}T^{A}_{A'}(x).\]
This transformation is not correct for a vector field. To amend this, the manifold $\C{N}$ must be split into a direct sum using a linear connection from which we may define new variables that transform correctly. Fix a linear connection $\nabla=\nabla_{\Pi E}\oplus\nabla_{\Pi F}$ on $\C{M}$ with connection coefficients $A^B_{iA}$ which transform as
    \[A^B_{iA} = \p_i\left(T^{A'}_A\right)T^B_{A'} + (-1)^{i(A + A')}T^{A'}_AJ^{i'}_iA^{B'}_{i'A'}T^B_{B'}.\]
Define new coordinates $x^*_i = X^*_i - (-1)^{iB}z^BA^A_{iB}z^*_A$ called \emph{long momenta} \cite{Vor02}. By direct computation we obtain the transformation law $x^*_i = J^{i'}_ix^*_{i'}$, which now transform with respect to the Jacobian matrix. In what follows we will use long momenta on our manifold $\C{N}$, which now splits into a direct sum
    \begin{equation*}
    \C{N} = \Pi E\oplus\Pi F\oplus E^*\oplus F^*\oplus \Pi T^*M.
    \end{equation*}
As well as the Grassmann grading from the parity of the coordinates and the momentum grading from the odd cotangent structure, $\C{N}$ carries two additional gradings. The first is a $\mathbb{Z}$-grading called \emph{the ghost number}, the name being carried over from the physical usage when introducing ghost variables. Since our manifold $M$ is assumed to be purely even, the ghost grading is related to the Grassmann parity by the parity equal to the ghost number modulo $2$. The second grading is an $\mathbb{N}$-grading called \emph{the resolution degree}. The resolution degree is an auxiliary grading which arises in the construction which will be detailed later. For clarity, we express the coordinates on $\C{N}$ in the following table together with their respective gradings:
    \begin{equation}\begin{tabular}{c|c|c|c|c|c|c}
   & $x^i$ & $\eta^a$ & $c^\alpha$ & $x^*_i$ & $\eta^*_a$ & $c^*_\alpha$ \\
  \hline
  Parity ($\epsilon$) & 0 & 1 & 1 & 1 & 0 & 0 \\
  Ghost ($\gh$) & 0 & -1 & 1 & 1 & 2 & 0 \label{momenta gradings}\\
  Momentum Degree ($\Deg$) & 0 & 0 & 0 & 1 & 1 & 1 \\
  Resolution ($\res$) & 0 & 1 & 0 & 0 & 0 & 1 \\
\end{tabular}\end{equation}
It will be convenient to introduce some collective notation $\phi^A = (x^i,\eta^a,c^\alpha)$ and $\phi^*_A = (x^*_i,\eta^*_a,c^*_\alpha)$. Notice from the table \eqref{momenta gradings} that there exists the relation
    \[\gh(\phi^*_A) = 1-\gh(\phi^A).\]
\begin{remark}
For the construction realising a Hamiltonian gauge theory, the variables $c^\alpha$ correspond to the standard BFV ghosts, and the $\eta^a$ corresponding to their ghost momenta. These variables are usually taken as conjugate variables with the canonical Poisson bracket on extended phase space, however we do not assume this. If the gauge theory is Lagrangian, then these correspond to the BV fields and anti-fields. In the presence of reducibility in the gauge system, ghosts for ghosts must be introduced as additional variables in the chains of vector bundles described in remark \ref{Remark reducibility in gt}. See \cite{HT92} for details.
\end{remark}

The manifold $\Pi T^*\C{M}$ comes equipped with a canonical odd Poisson bracket induced from the canonical odd symplectic structure $d\left(dx^iX^*_i + d\eta^a\eta^*_a + dc^\alpha c^*_\alpha\right)$. The use of long momenta associated with the connection $\nabla$ twists this odd bracket so that it has the following expressions in local coordinates:
    \begin{equation}\begin{array}{lcl}
   (x^*_i,c^\alpha)=c^\beta A^\alpha_{\beta i}\,, \qquad & (x^*_i,c^*_\alpha)=A^\beta_{i\alpha}c^*_\beta\,, \qquad & (\eta^*_a,\eta^b)=\delta^b_a\,,\\
   (x^*_i,\eta^a)=\eta^bA^a_{bi}\,, \qquad & (x^*_i,\eta^*_a)=A^b_{ia}\eta^*_b\,, \qquad & (c^*_\alpha,c^\beta)=\delta_\alpha^\beta\,,\label{Poisson bracket}\\
   (x^*_i,x^j)=\delta_i^j\,, \qquad & (x^*_i,x^*_j)=c^\beta\C{R}^\alpha_{\beta ij}c^*_\alpha  + \eta^a\C{R}^b_{aij}\eta^*_b\,, \qquad & \end{array}\end{equation}
and where all other brackets vanish identically. Here the terms $\C{R}^b_{aij}$ and $\C{R}^\alpha_{\beta ij}$ are the components of the curvatures of the connections $\nabla_{\Pi E}$ and $\nabla_{\Pi F}$ respectively. This is an odd Poisson bracket of ghost degree $-1$, which corresponds to the odd symplectic $2$-form $d\Theta$, where
    \[\Theta = dx^ix^*_i + \nabla_{\Pi E}\eta^a\eta^*_a + \nabla_{\Pi F}c^\alpha c^*_\alpha,\]
and
    \[\nabla_{\Pi E}\eta^a = d\eta^a - \eta^bdx^iA^a_{ib}(x), \qquad \nabla_{\Pi F}c^\alpha = dc^\alpha - c^\beta dx^iA^\alpha_{i\beta}(x).\]
Notice that in the case when the bundles $E$ and $F$ are trivial, all connection components and curvatures vanish, and we recover the canonical odd Poisson bracket on the odd cotangent bundle $\Pi T^*\C{M}$. In all cases, $\C{M}\subset\C{N}$ is a Lagrangian submanifold defined by setting all odd momenta $\phi^*$ to be zero.

\subsection{The Master Function}
Introduce $S\in C^\infty(\C{N})$ with the gradings
    \[\epsilon(S) = 0\,,\qquad \gh(S)=2\,,\qquad  \Deg(S)>0,\]
such that $S$ self commutes under the Poisson bracket \eqref{Poisson bracket},
    \begin{equation}\label{master equation}(S,S) = 0.\end{equation}
This function $S$ will be called the \emph{master function} which satisfies the \emph{master equation} \eqref{master equation}. The master function will encode all the information about the sections $T$ and $R$, the compatibility conditions between them, and the weak Poisson structure introduced. Locally, $S$ has the appearance
    \[S = T^a\eta^*_a + c^\alpha R^i_\alpha x^*_i  + P^{ij}x^*_jx^*_i + \left(c^\beta c^\alpha U^\gamma_{\alpha\beta} + V^{\gamma ij}x^*_jx^*_i + c^\alpha W^{\gamma i}_\alpha x^*_i + Y^{\gamma a}\eta^*_a\right)c^*_\gamma +\]
    \[\eta^a\left(c^\beta c^\alpha A^{i}_{\alpha\beta a}x^*_i + c^\alpha B^{ij}_{\alpha a} x^*_jx^*_i + D^{ijk}_ax^*_kx^*_jx^*_i + c^\alpha E^b_{\alpha a}\eta^*_b + F^{ib}_a\eta^*_bx^*_i\right) + \res\geq 2 \mbox{ terms}.\]
The coefficients in the terms of resolution degree zero may be identified with the components $T^a$, the vector fields $R_\alpha$ and the components of the weak Poisson bivector $P^{ij}$. The other functions appearing as coefficients are the higher structure functions and give higher relations between the lower degree terms. In order to obtain the regularity conditions imposed previously, we need to assume that
    \[\rank\left.\left(\F{\p^2S}{\p\phi^A\phi^*_A}\right)\right|_{dS=0} = (n,m),\]
where $m$ and $n$ are the ranks of the bundles $E$ and $F$ respectively.

Consider an expansion of $S$ in terms of the momentum degree
    \[S = \sum^\infty_{k=1}S^k = Q^A\phi^*_A + \Pi^{AB}\phi^*_B\phi^*_A + \Xi^{ABC}\phi^*_C\phi^*_B\phi^*_A +\cdots.\]
The master equation gives the following relations in the lowest degrees:
    \begin{equation}\label{master equation equations}(Q,Q) = 0\,, \qquad (Q,\Pi)=0\,, \qquad \mbox{and}\qquad (\Pi,\Pi) = -2(Q,\Xi).\end{equation}
The first of these ensures that $Q = Q^A(\phi)\phi^*_A$ is a homological vector field on the Lagrangian submanifold $\C{M}$; $Q$ is a Grassmann odd, ghost $+1$ vector field that squares to zero. The first few terms in local coordinates are
    \begin{equation}\label{Q}
    Q = T^a\F{\p}{\p\eta^a} + c^\alpha R^i_\alpha\left(c^\beta A^{\gamma}_{\beta i}\F{\p}{\p c^\gamma} + \eta^bA^{a}_{bi}\F{\p}{\p \eta^a} + \F{\p}{\p x^i}\right) + c^\beta c^\alpha U^\gamma_{\alpha\beta}\F{\p}{\p c^\gamma}+\cdots.
    \end{equation}
It is the  restriction of the Hamiltonian derivation defined by $S$ to the submanifold $\C{M}$, $(S,-)|_{\C{M}} = Q$.
Notice that $Q$ contains all the information about the components of the section $T$ and the vector fields $R_\alpha$. It is the component of $S$ which encodes the construction outlined in the previous section, independent of the Poisson structure which is present in the component $\Pi^{AB}$.

The existence and uniqueness of the master function $S$ arise as a solution to \eqref{master equation} subject to certain boundary conditions. These boundary conditions are precisely restrictions on the resolution degree coming from the extended Koszul-Tate resolution set into the background of these calculations. The existence of solutions to \eqref{master equation} follows from standard homological perturbation theory, the details of which may be found exactly in \cite{HT92}, or more generally in \cite{Tat57}. Here, only a sketch of the existence of $S$ will be provided.

Consider an expansion of $S$ in terms of the resolution degree
    \[S = \sum^\infty_{n=0}S_n.\]
The first terms in $S$ are known due to the nature of the construction. They consist of the Koszul-Tate differential $\delta$, a differential $d$, and a correction term $\tilde{s}$ such that $\delta^2=0$, $d^2=[\delta,\tilde{s}]$ and the terms of resolution degree $-2$, $-1$ and $0$ in $(S,S) = 2S^2$ vanish. The Koszul-Tate differential is a differential providing a homological resolution of the algebra $H_0(\delta) = C^\infty(\Sigma)$. That is, it implements the restriction to the surface $\Sigma$ on the level of homology. It does so by setting
    \[(\ker\delta)_0 = C^\infty(M), \qquad (\im\delta)_0 = \langle T^a\rangle.\]
It is a differential of resolution degree $-1$, ghost degree $+1$ and is required to kill anything with resolution degree $0$. These properties follow from the homological perturbation requirements. Using these, $\delta$ takes the form
    \begin{equation}\label{delta differential}
    \delta = T^a\F{\p}{\p\eta^a} + R^i_\alpha x^*_i\F{\p}{\p c^*_\alpha}.
    \end{equation}
The differential $d$ is called the longitudinal differential and is a differential along leaves of the foliation; to the vectors $R_\alpha$ correspond dual differential $1$-forms forming the dual distribution. The longitudinal differential is a differential on these forms. It implements the concept of invariant functions over the integral submanifolds by annihilating those invariant functions. These terms are all known from the information contained in the initial data. Therefore in order to construct $S$ we need to consider the higher terms in resolution degree $k$. In fact, we assume the construction for $k-1$ and look at the conditions on $k$. Substituting the sum $S' = \sum^k_{n=0}S_n$ into \eqref{master equation} we come to the equation
    \[\delta S_k = \rho_{k-1}(S_0,\ldots, S_{k-1}),\]
where $\rho_{k-1}$ is a function of resolution degree $k-1$. One can then eliminate $\rho_{k-1}$ by a suitable choice of $S_k$. From the Jacobi identity $((S,S),S)\equiv 0$, the term $\rho_{k-1}$ is $\delta$-closed by comparing resolution degrees. The assumed regularity conditions ensure that the differential $\delta$ is acyclic
    \[H_k(\delta)=0,\qquad k>0.\]
Therefore $\rho_{k-1}$ is not only $\delta$-closed, it is $\delta$-exact. So there exists $S_k$ satisfying this condition. The other terms in $S$ are found recursively from the previous terms. We need only to check the first equation to finish the sketch of the proof,
    \[\delta S_1 = \rho_0(S_0) = (S_0,S_0).\]
This gives equations that are equivalent to the defining relations for the section $T$, the homomorphism $R$, and the bivector field $P$.

Now $S$ is unique up to a canonical transformation. The ambiguity in the construction above, replacing
    \[S_k \mapsto S_k+\delta F_{k+1},\]
may be absorbed by a canonical transformation of the odd Poisson manifold $\C{N}$. This is given explicitly in \cite{KLS05} amongst others.

\subsection{Cohomology of $Q$}
The master function gives rise to  the homological vector field $Q$ on the Lagrangian submanifold $\C{M}\subset\C{N}$. This may be viewed as a differential which provides the algebra $C^\infty(\C{N})$ with the structure of a cochain complex $Q:C^{\infty,k}_l(\C{N})\rightarrow C^{\infty,k}_{l+1}(\C{N})$, naturally bi-graded by the momentum and ghost degrees, with the vector field
    \[QF = (S^1,F) \qquad \mbox{ for }\quad F\in C^\infty(\C{N}).\]
The cohomology groups for $Q$ decompose with respect to both the momentum degree $k$ and the ghost degree $l$,
    \[H(Q) = \bigoplus_{k,l}H^k_l(Q).\]

\begin{lemma}
For $k>l$, the groups $H^k_l(Q)$ are trivial.
\end{lemma}
\begin{proof}
First, $Q$ may be graded by resolution degree and decomposed into $Q = \delta + \Delta$, where
    \[\res(\delta) = -1,\qquad \res(\Delta) \geq0.\]
(Note that $\delta$ here is the differential \eqref{delta differential} restricted to the extended manifold $\C{M}$.) Further, we may grade $\Delta$ by resolution degree so that
    \[Q = \delta + \Delta_0 + \sum^\infty_{i=1}\Delta_i.\]
Let us choose $F\in H^k_l(Q)$ such that $k>l$. Comparing the degrees in table \eqref{momenta gradings}, it can be seen that $\res(\phi^{(*)})\geq\Deg(\phi^{(*)})-\gh(\phi^{(*)})$, so that
    \[\res(F)>0.\]
Suppose that the lowest resolution degree term in $F$ has degree $r$. Grade $F$ as $F=\sum^\infty_{n=r}F_n$ (starting from $n=r$), with $\res(F_n)=n$. We wish to find $G\in C^{\infty,k}_{l-1}(Q)$ such that $F = QG$. Note that $\res(G)>\res(F)$.

As $F$ is a $Q$-cocycle, $QF=0$ and in resolution degree $r-1$, we have
    \[\delta F_r=0\Rightarrow F_r = \delta G_{r+1}\]
for some $G_{r+1}$ since $\delta$ is acyclic in positive resolution degree. Looking at the second lowest degree, resolution degree $r$, we come to
    \begin{equation}\label{homological pert}
    \delta F_{r+1} + \Delta_0F_r = 0.
    \end{equation}
From $Q^2F\equiv0$ we obtain
    \[\delta(\Delta F) + \Delta(\delta F) + \Delta^2F\equiv 0.\]
On expanding this in resolution degree, we come to the equations
    \begin{equation}\label{first equation}
    \delta\Delta_0 + \Delta_0\delta = 0,
    \end{equation}
    \[\delta\Delta_1 + \Delta_1\delta + \Delta^2_0=0,\]
plus equations of resolution degree greater than $0$. Since $F_r = \delta G_{r+1}$, equation \eqref{homological pert} becomes
    \begin{align*}
    \delta F_{r+1} & = -\Delta_0F_r\\
    & = -\Delta_0(\delta G_{r+1})\\
    & = \delta(\Delta_0 G_{r+1}) \qquad \mbox{ by } \eqref{first equation}.\end{align*}
Therefore
    \[\delta(F_{r+1} - \Delta_0 G_{r+1}) = 0 \Rightarrow F_{r+1} - \Delta_0 G_{r+1} = \delta G_{r+2},\]
for some function $G_{r+2}$. Finally,
    \[F_{r+1} = \Delta_0 G_{r+1} + \delta G_{r+2} = QG|_{\res(QG) = r+1}.\]
Continually solving gives a function $G$ such that $QG=F$ and so all cohomology groups are trivial if $k>l$.
\end{proof}

\begin{proposition}\label{cocycle projectible condition}
A multivector field $U$ is projectible if and only if it is a $Q$-cocycle.
\end{proposition}

\begin{proof}
Let $F$ be a multivector field with $\Deg F=k$ which is extended homogeneously by ghost degree $k$ terms. In resolution degree,
    \[F = \sum^\infty_{r=0}F_r,\]
noting that $\res(F)\geq0$ by the observation in Lemma 1. Looking at fixed resolution degree, the term $QF$ gives the equations
    \[\delta F_{r+1} = \rho_{r+1}(F_0,\ldots,F_r).\]
Suppose that $QF=0$. Then the equation in resolution degree $0$ is satisfied if and only if $F$ is projectible, that is $\llbracket F, \ff{I}\rrbracket\in\ff{I}$.
Conversely, if $F$ is projectible, the first equation is satisfied in resolution degree $0$. By a similar argument to the proof of Lemma 1, all higher resolution terms are $Q$-coboundaries due to the condition of $Q^2F\equiv0$ and the acyclicity of $\delta$ in positive resolution degrees. The projectibility conditions are required in this case to start the induction due to the fact that $\delta$ is not acyclic in degree $0$. Therefore all projectible multivectors may be lifted to a $Q$-cocycle. Two equivalent multivectors, in the sense of the relation  \eqref{equivalence of multivectors}, differ by a $Q$-coboundary.
\end{proof}
From the construction of $S$, the $k$-th cohomology $H^k(Q)$ of the differential $Q$ coincides with the $k$-th cohomology group $H^k(d|H_0(\delta))$, the cohomology of $d$ in the homology group $H_0(\delta)$. In particular, the first cohomology group $H^0_0(Q)$ in ghost degree $0$ is isomorphic to the algebra of physical observables $C^\infty(N)$, \cite{KLS05, LS05}. Intuitively, $H^0_0(Q)$ contains the classes of functions restricted to the manifold $\Sigma$ by $\delta$ which are constant along the integral submanifolds of the foliation given by the vector fields $R_\alpha$. The invariance is shown by the annihilation of these by the differential $d$. Other notable cohomology groups include $H^1_1(Q)$ and $H^2_2(Q)$ which consist of the all projectible vector fields and bivector fields  on $M$ respectively under restriction to $\Pi T^*M$.

The master function $S$ generates a sequence of higher Poisson brackets on the manifold $\C{M}$ which provides $C^\infty(\C{M})$ with the structure of a homotopy Poisson algebra, or a $P_\infty$-algebra. Define the $k$-bracket as
    \begin{equation}\label{infinity poisson structure}
    \{F_1,\ldots,F_k\}_S := \left.\big(\cdots(S,F_1),\ldots,F_k\big)\right|_{\C{M}},
    \end{equation}
as a sequence of nested brackets, for functions $F_1,\ldots,F_k\in C^\infty(\C{M})$. Each bracket is a derivation with respect to each argument and the Grassmann parity of a $k$-bracket is equal to $k$ mod $2$. Since $(S,S)=0$, the higher Jacobi identities hold for all higher brackets, \cite{Vor05}. These higher Jacobi identities are Jacobi identities satisfied up to higher homotopies. In the case of the $3$-bracket, this is a genuine homotopy with $Q$ and the trilinear bracket:
    \begin{align*}
    \{\{F,G\},H\} + (-1)^{GH}\{\{F,H\},G\}+ & (-1)^{F(G+H)}\{\{G,H\},F\} =\\
 Q\{F,G,H\} + & \{QF,G,H\}+(-1)^{FG}\{QG,F,H\}+(-1)^{H(F+G)}\{QH,F,G\}.\end{align*}
The usual Jacobi identity is satisfied up to $Q$-coboundaries. Therefore, on passing to the $Q$-cohomology, the function $S$ induces a genuine even Poisson bracket on the space $H^0_\bullet(Q)$. This is exactly the Poisson structure induced from the weak Poisson bracket on the leaf space $N$, but with the extension to the extended manifold $\C{M}$.


\section{A Lift to the Algebra of Forms}
\subsection{Lifting the Master Function}
In \cite{KV08}, a construction was given to canonically lift any Poisson bracket or sequence of Poisson brackets on a manifold to an odd Poisson bracket or sequence of odd Poisson brackets on the odd tangent bundle. In the case of an even binary bracket, this lift reproduces the odd Koszul bracket on differential forms. To do this, a function on the cotangent bundle is canonically identified using the homological vector field corresponding to the Poisson structure. Our approach will be slightly different, in that we will lift the master function $S$ to the odd tangent bundle $\Pi T\C{N}$ using a secondary canonical identification. On this space, introduce the odd velocities to their coordinates on $\C{N}$, with gradings as in the following table:
\[\begin{tabular}{c|c|c|c|c|c|c|c|c|c|c|c|c}
   & $x^i$ & $\eta^a$ & $c^\alpha$ & $x^*_i$ & $\eta^*_a$ & $c^*_\alpha$ & $dx^i$ & $d\eta^a$ & $dc^\alpha$ & $dx^*_i$ & $d\eta^*_a$ & $dc^*_\alpha$ \\
  \hline
  Parity & 0 & 1 & 1 & 1 & 0 & 0 & 1 & 0 & 0 & 0 & 1 & 1 \\
  Ghost & 0 & -1 & 1 & 1 & 2 & 0 & -1 & -2 & 0 & 0 & 1 & -1 \\
  Momentum Deg & 0 & 0 & 0 & 1 & 1 & 1 & - & - & - & - & - & - \\
  Res Deg & 0 & 1 & 0 & 0 & 0 & 1 & - & - & - & - & - & - \\
\end{tabular}\]
Collectively, let $d\phi^A = (dx^i,d\eta^a,dc^\alpha)$ and $d\phi^*_A = (dx^*_i,d\eta^*_a,dc^*_\alpha)$. This manifold may be naturally identified with $T^*(\Pi T\C{M})$ by the diffeomorphism
    \[\kappa:\Pi T(\C{N})\mapsto T^*(\C{N}^*),\]
which can be seen as an analogue to the Tulczyjew isomorphism, \cite{Tul77} (or see \cite{MX94}). It is a composition of two natural identifications,
    \[\Pi T(\C{N}) \cong T^*(\C{N})\quad \mbox{ and } \quad T^*(\C{N})\cong T^*(\C{N}^*).\]
The first is the canonical pairing of the odd tangent bundle with the cotangent bundle using the canonical odd symplectic form on $\Pi T^*\C{M}$. The second is the canonical isomorphism which interchanges fibre coordinates with their corresponding momenta. This isomorphism was first described in the work of Mackenzie and Xu, and may be found in \cite{MX94}, which was then extended to supermanifolds by Voronov, \cite{Vor02}. (In \cite{Vor02} the expression gives a symplectomorphism whereas we prefer to take the choice of signs which gives an anti-symplectomorphism between the two canonical structures.) In terms of local coordinates $(\varphi^A,d\varphi^A,p_A,\pi_A)$ on $T^*(\C{N}^*)$, $\kappa$ is described as
    \[\kappa^*(\varphi^A) = \phi^A ,\quad \kappa^*(d\varphi^A) = (-1)^Ad\phi^A,\quad \kappa^*(p_A) = d\phi^*_A,\quad \kappa^*(\pi_A) =-\phi^*_A .\]
Under $\kappa$, the symplectic form $\alpha = d\pi_Ad(d\varphi^A) + dp_Ad\varphi^A$ is given by $\kappa^*\alpha = d(d\phi^*_A)d\phi^A - d(d\phi^A)d\phi^*_A$, which generates the even non-degenerate Poisson bracket of ghost degree $0$:
    \[\{F,G\} = (-1)^{A(F+1)}\left(\F{\p F}{\p d\phi^*_A}\F{\p G}{\p \phi^A} + (-1)^{F}\F{\p F}{\p d\phi^A}\F{\p G}{\p \phi^*_A}\right) - (-1)^{AF}\left(\F{\p F}{\p \phi^A}\F{\p G}{\p d\phi^*_A }-(-1)^F\F{\p F}{\p \phi^*_A}\F{\p G}{\p d\phi^A} \right).\]
Define the function $\Delta\in C^\infty(\Pi T\C{N})$ as the canonical Grassmann odd, ghost $-1$ function $\Delta = (-1)^Ad\phi^Ad\phi^*_A.$ Notice that $\Delta$ self-commutes non-trivially; $\{\Delta,\Delta\}=0$. Using this, the master function $S$ may be lifted to the odd tangent bundle,
    \[\Psi = \{\Delta,S\}\]
as an odd function of ghost degree $+1$.  Explicitly, for selected terms,
    \[\Psi = dx^i\left(\p_iT^a\eta^*_a + c^\alpha\p_iR^j_\alpha x^*_j + \p_iP^{jk}x^*_kx^*_j\right) - dc^\alpha R^i_\alpha x^*_i\]
    \[+\left(c^\alpha R^i_\alpha + 2P^{ij}x^*_j + \eta^ac^\beta c^\alpha A^{i}_{\alpha\beta a}\right)dx^*_i+ T^ad\eta^*_a + c^\beta c^\alpha U^\gamma_{\alpha\beta}dc^*_\gamma +\cdots.\]

The function $\Psi$ is Poisson-nilpotent if and only if $S$ satisfies the master equation \eqref{master equation}. Indeed, with the even bracket $\{-,-\}$, $\Delta$ generates the odd bracket $(-,-)$ on $\C{N}$ by the formula $(-,-):= \left.\big\{\{\Delta,-\},-\big\}\right|_{\C{N}}$. Therefore, if $\Psi$ Poisson commutes it follows from the Jacobi identity
    \begin{equation}\label{Omega comm}
    0 = \{\Psi,\Psi\} = \big\{\{\Delta,S\},\{\Delta,S\}\big\} = \Big\{\big\{\{\Delta,S\},\Delta\big\},S\Big\}-\Big\{\Delta,\big\{\left\{\Delta,S\right\},S\big\}\Big\},
    \end{equation}
that \eqref{master equation} is satisfied, since the first term on the right hand side vanishes from the nilpotency of $\Delta$.

In the same way that multivector fields on a manifold are identified with functions on its odd cotangent bundle, differential forms may be identified with functions on the odd tangent bundle. Identify differential forms on the extended manifold $\C{M}$ with fibrewise polynomial functions $\Omega(\C{M})\subset C^\infty(\Pi T\C{M})$. Under the identification $\kappa$, they are those functions which are fibrewise polynomial on the Lagrangian submanifold $\C{L}\subset \Pi T\C{N}$,
    \[\C{L} = \left.\big\{(\phi^A,\phi^*_A,d\phi^A,d\phi^*_A)\in\Pi T\C{N}\,\,\right|\,\,\phi^* = 0 = d\phi^*\big\}.\]
Specifically, for $\omega\in\Omega(\C{M})$, $\kappa^*\omega = \kappa^*\omega(\phi,d\phi)$. We will freely identify these functions and write simply $\omega\in\Omega(\C{M})$ for a function $\kappa^*\omega\in\kappa^*(\Omega(\C{M}))$. (Notice that unlike multivector fields however we do not need to introduce new coordinates via a connection, since the $d\phi$ transform correctly.)

Since $S$ satisfies the master equation, $\Psi$ defines a homological vector field on $\C{L}$, given as the restriction to $\C{L}$ of the derivation defined by $\Psi$: $\hat{Q} = \{\Psi,-\}|_{\C{L}}$.
Locally,
    \begin{equation}\label{Q-exp}
    \hat{Q} = \left(c^\alpha R^i_\alpha  + \eta^ac^\alpha c^\beta A^{i}_{\beta\alpha a}\right)\F{\p}{\p x^i} + c^\beta c^\alpha U^\gamma_{\alpha\beta}\F{\p}{\p c^\gamma} + T^a\F{\p}{\p\eta^a} +dx^i\p_iT_a\F{\p}{\p d\eta_a}\end{equation}
    \[+ \left(c^\alpha dx^i\p_iR^j_\alpha + dc^\alpha R^j_\alpha\right)\F{\p}{\p dx^j} + c^\beta c^\alpha U^\gamma_{\alpha\beta}\F{\p}{\p dc^\gamma} +\cdots.\]
Since $\{\Psi,\Psi\}=0$, $\Psi$ generates a sequence of derived brackets on the Lagrangian submanifold $\C{L}$ by
    \begin{equation}\label{odd koszul-type brackets}
    [\omega_1,\ldots,\omega_k]:=\left.\big\{\cdots\{\Psi,\omega_1\},\ldots,\omega_k\big\}\right|_{\C{L}},
    \end{equation}
for functions $\omega_i\in C^\infty(\C{L})$. This sequence of odd brackets provides the algebra $C^\infty(\C{L})$ with the structure of an $S_\infty$-algebra or an odd homotopy Poisson algebra. This is precisely the $S_\infty$-structure defined in \cite{KV08} on the algebra of differential forms, corresponding to a $P_\infty$-structure on the base manifold $\C{M}$; the $P_\infty$-structure in our case is defined by the sequence of brackets \eqref{infinity poisson structure}.

\subsection{The Cohomology of $\hat{Q}$}
Analogous to the case of $Q$, the vector field $\hat{Q}$ acts as a differential of ghost degree $+1$ on $C^\infty(\C{L})$. This turns the algebra into a complex $\hat{Q}:C^\infty_l(\C{L})\rightarrow C^\infty_{l+1}(\C{L})$ naturally graded by ghost degree $l$. Since all momentum terms are zero over $\C{L}$, it is beneficial to use the natural grading on the algebra of forms pulled back to $\C{L}$. We call this grading the \emph{form degree}, with the terms $d\phi^A$ having form degree $+1$, and all other variables are assigned zero. The cohomology of the differential $\tilde{Q}$ decomposes with respect to the form degree $k$ and the ghost degree:
    \[H(\hat{Q}) = \bigoplus_{k,l}H^k_l(\hat{Q}).\]
Unlike the algebra of multivector fields on $\C{M}$ which has only the odd Poisson bracket \eqref{Poisson bracket}, the algebra of differential forms on $\C{M}$ inherits the whole sequence of odd Poisson brackets given by \eqref{odd koszul-type brackets}. As such, we obtain a sequence of odd Poisson brackets when passing to the cohomology of $\hat{Q}$. In particular, there is a true odd binary bracket on ($\hat{Q}$-cohomology classes of) differential forms, which may be viewed as a direct lift (in the sense of \cite{KV08}) of the Poisson bracket that was observed on the cohomology $H^0(Q) = \oplus_lH^0_l(Q)$. Such a binary bracket may be called the Koszul bracket of differential forms on $N$, and the binary bracket on $C^\infty(\C{L})$ restricts to a weak Koszul-type bracket on differential forms on $M$. To parallel Proposition \eqref{cocycle projectible condition} then, we make the following definition.

\begin{definition}
A \emph{projectible differential form on $M$} is a function $\omega\in C^\infty(\Pi TM)$ such that its appropriate extension to the Lagrangian submanifold $\C{L}$ is a $\hat{Q}$-cocycle.
\end{definition}

Indeed, for a form to be invariant over the integral submanifolds defined by the vector fields $R_\alpha$, it is natural to ask for the condition $\C{L}_{R_\alpha}\omega \propto T + dT$; that the Lie derivative be proportional to the components $T^a$ of the section $T$ and their differentials. With expansion (\ref{Q-exp}), one can see that this condition is reproduced exactly by the cocycle equation $\hat{Q}\omega=0$ in the lowest degrees.


\begin{proposition}\label{prop2}\hfill

\begin{enumerate}
  \item Under the odd binary bracket, projectible forms form a closed Poisson subalgebra.
  \item This subalgebra is closed under $d$, the exterior differential.
  \item The algebra of projectible differential forms is a module over the algebra of projectible vector fields with respect to the interior product.
\end{enumerate}
\end{proposition}

\begin{proof}Let $\omega$,$\tau\in C^\infty(\C{L})$ be $\hat Q$-cocycles; $\hat{Q}\omega = 0 = \hat{Q}\tau$.\hfill
\begin{enumerate}
\item  The first claim is obvious by an application of the Jacobi identity:
    \begin{align*}
    \tilde{Q}[\omega,\tau] & = \left.\Big\{\Psi, \big\{\{\Psi,\omega\},\tau\big\}|_{\C{L}}\Big\}\right|_{\C{L}}\\
    & = \left.\Big\{\big\{\Psi,\{\Psi,\omega\}\big\},\tau\Big\}\right|_{\C{L}} - (-1)^\omega \left.\Big\{\{\Psi,\omega\},\{\Psi,\tau\}\Big\}\right|_{\C{L}}\end{align*}
The first term vanishes since $\{\Psi,\Psi\} = 0$. The second term requires both the projectibility of $\omega$ and $\tau$ to vanish when restricted to the Lagrangian submanifold $\C{L}$.

\item The de Rham differential is given by the vector field $d\varphi^A\F{\p}{\p\varphi^A}$ on $\Pi T\C{M}$. Pulling back to $\C{L}$, the vector field $d$ has the local expression
    \[d = (-1)^Ad\phi^A\F{\p}{\p\phi^A},\]
and is Hamiltonian since $d = \{\Delta,-\}|_\C{L}$. Then
    \begin{align*}
    d\left(\hat{Q}\omega\right) & = \left.\big\{\Delta,\{\Psi,\omega\}|_{\C{L}}\big\}\right|_{\C{L}}\\
    & = \left.\big\{\{\Delta,\Psi\},\omega\big\}\right|_{\C{L}} - \left.\big\{\Psi,\{\Delta,\omega\}\big\}\right|_{\C{L}}\\
    & = -\hat{Q}\left(d\omega\right),\end{align*}
since $\{\Delta,\Delta\}=0$. So if a form is projectible, its exterior differential is also. (Notice that the commutator $[\hat Q,d] = \hat Qd + d\hat Q = 0$.)

\item Let $X$ be a projectible vector field, written $X^A(\phi)\phi^*_A$ as a function on $C^\infty(\Pi T\C{N})$. The interior product with $X$ is the vector field
    \[\imath_X = (-1)^XX^A(\phi)\F{\p}{\p d\phi^A}.\]
which may be expressed in terms of the Poisson bracket:
    \[(-1)^X\{X,-\}|_{\C{L}} = (-1)^XX^A(\phi)\F{\p}{\p d\phi^A} = \imath_X.\]
Consider
    \begin{align*}
    \hat Q\left(\imath_X\omega\right) & = (-1)^X\left.\big\{\Psi,\{X,\omega\}|_{\C{L}}\big\}\right|_{\C{L}}\\
    & = (-1)^X\left.\big\{\{\Psi,X\},\omega\big\}\right|_{\C{L}} + \left.\big\{X,\{\Psi,\omega\}\big\}\right|_{\C{L}}.\end{align*}
It can be seen that both terms vanish, the first from the projectibility of $X$, and the second from the projectibility of $\omega$. So $\imath_X\omega$ defines a projectible form so long as $X$ remains a projectible vector field.
\end{enumerate}
\end{proof}

\subsection{One Parameter Subgroups of Automorphisms}
Now we turn to one parameter subgroups of automorphisms of projectible differential forms on $M$. The interest here is physical. In line with Remark \ref{dynamics remark}, a projectible vector field provides the related gauge system with dynamics, and the one parameter subgroup generated by the flow gives the evolution of measurable observables through time. As such, given a projectible vector field, we can describe the evolution of differential forms over the flow of $V$, and ask about those forms which are constant.

As in \cite{LS05}, introduce a vector field $V$ on the space of leaves $N$ which can be seen to provide dynamics. Such a vector field is a projectible vector field on $M$ and so its appropriate extension to $\C{M}$ is a $Q$-cocycle, but further it is required to satisfy its own master equation
    \begin{equation}\label{master equation2}(S,V)=0.\end{equation}
This condition defines $V$ as a weak Poisson vector field on $M$ equipped with the weak Poisson bracket.
The vector field is a function $V\in C^\infty(\C{N})$ and is graded as
    \[\gh(V) = +1, \qquad \epsilon(V) = +1,\qquad \Deg(V)>0.\]
Locally, up to resolution degree $1$, $V$ has the expression
    \[V = V^ix^*_i + \left(c^\beta W^\alpha_\beta + G^{\alpha i}x^*_i\right)c^*_\alpha + \eta^a\left(M^b_a\eta^*_b + L^{ij}_ax^*_jx^*_i + c^\alpha N^{i}_{\alpha a}x^*_i\right) + \cdots.\]
That $V$ exists is proved in the same way as for $S$. The higher terms are found by recursive solutions to \eqref{master equation2} in fixed resolution degrees.

For a projectible function (physical observable) $F\in H^0_0(Q)$, the Lie derivative of $F$ along $V$ was given by \eqref{time evolution}:
    \[\dot{F} = \left.(V,F)\right|_{\phi^*=0}.\]
The Lie derivative of a differential form along a projectible vector field may also be defined. In a similar way to $S$, $V$ must be lifted to an even function on the odd tangent bundle, denoted by $\Gamma\in C^\infty(\Pi T\C{N})$, where
    \[\Gamma = \{\Delta, V\}.\]
We consider the analogous relation to \eqref{master equation2}:
    \begin{align}
    \{\Psi,\Gamma\} & = \big\{\{\Delta,S\},\{\Delta,V\}\big\}\nonumber\\
    & = \big\{\{\Psi,\Delta\},V\big\} - \big\{\Delta,\{\{\Delta,S\},V\}\big\}.\label{Gamma comm}\end{align}
The first term vanishes from the Poisson-nilpotency of $\Delta$, and since $V$ satisfies the master equation, the second term vanishes also. So $\{\Psi,\Gamma\}=0$ if and only if the master equation \eqref{master equation2} is satisfied.

A projectible vector field may be lifted to a $\Psi$-commuting function and we may define the Lie derivative of a projectible differential form along $V$ as
    \begin{equation}\label{evolution forms}
    \dot{\omega} = \{\Gamma,\omega\}|_{\C{L}}.
    \end{equation}
That $\dot\omega$ remains projectible can be seen by an application of the Jacobi identity; it will remain projectible so long as $V$ remains projectible. If the form is constant over the flow generated by $V$, then $\dot\omega = 0$ at the level of cohomology. Therefore the form $\dot\omega$ is $\hat Q$-exact, and so is a $\hat Q$-coboundary of some (not necessarily projectible) form $\tau$ of ghost degree one less than $\omega$. In terms of the original manifold $M$, this reproduces the condition
    \[\C{L}_V\omega|_\Sigma = 0: \qquad \C{L}_V\omega \propto T + dT\]
in the correct degrees. For example, let $\omega = d\phi^A\omega_A(\phi)$ be a $1$-form of ghost degree $-1$. An appropriate form $\tau$ of ghost degree $-2$ has expression $\tau = dx^i\eta^a\tau_{ai} + d\eta^a\tau_a +\cdots$. Performing the calculations we come to
    \[V^i\p_i\omega_jdx^j + dx^j\p_jV^i\omega_i = T^adx^i\tau_{ia} + dx^i\p_iT^a\tau_a,\]
the expression above but in local coordinates.

The natural operations on forms have the following degrees:
    \[d:H^{k}_{l}(\hat{Q})\rightarrow H^{k+1}_{l-1}(\hat{Q})\,,\qquad \imath_V:H^{k}_{l}(\hat{Q})\rightarrow H^{k-\Deg V}_{l+n}(\hat{Q}),\]
for $V$ a projectible homogeneous $\Deg(V)$ multivector field with $\gh(V) = +n$. Notice that if $V$ has degrees $\Deg(V) = +1$, $\gh(V)=+1$ ($V$ is a ghost $+1$ vector field), then
    \[\C{L}_V:H^{k}_{l}(\hat{Q})\rightarrow H^{k}_{l}(\hat{Q}).\]
The Lie derivative along a vector field $V$ is given by the well-known formula $\C{L}_V = d\circ \imath_V + (-1)^V\imath_V\circ d$. Since both $d$ and $\imath_V$ can be expressed in terms of the Poisson bracket, the Lie derivative can also, and by the Jacobi identity is equal to the expression \eqref{evolution forms}:
    \begin{equation}\label{Lie derivative}
    \C{L}_V\omega = (-1)^V\left.\big\{\{\Delta,V\},\omega\big\}\right|_{\C{L}} = (-1)^V\dot\omega.
    \end{equation}
This expression provides the following simply observation: if $\dot\omega = 0$ then $\dot{d\omega} = 0$ also.
\begin{proposition}
Let $V$ be a projectible vector field and $\Gamma$ its lift by $\Delta$. Suppose that $\omega$ is a projectible differential form and $X$ is a second projectible vector field. Then
    \[\dot{(\imath_X\omega)} = \imath_X(\dot\omega) + (-1)^X\left.\big\{(V,X),\omega\big\}\right|_{\C{L}},\]
where $(V,X)$ is the bracket \eqref{Poisson bracket} of vector fields $X$ and $V$. In terms of the Lie derivative,
    \[\C{L}_V\imath_X\omega = \imath_X\C{L}_V\omega + (-1)^{X+V}\left.\big\{(V,X),\omega\big\}\right|_{\C{L}}.\]
\end{proposition}

\begin{proof}
By Proposition \ref{prop2}, $\imath_X\omega$ is projectible. Now the statement follows from another application of the Jacobi identity and \eqref{Lie derivative}:
    \begin{align*}
    \dot{(\imath_X\omega)} & = (-1)^X\left.\big\{\Gamma,\{X,\omega\}\big\}\right|_{\C{L}}\\
     & = (-1)^X\left.\big\{\{\Gamma,X\},\omega\big\}\right|_{\C{L}} + (-1)^X\left.\big\{X,\{\Gamma,\omega\}\big\}\right|_{\C{L}}\\
     & = (-1)^X\left.\big\{(V,X),\omega\big\}\right|_{\C{L}} + \imath_X(\dot{\omega}).\end{align*}
\end{proof}
The interior product with $X$ commutes with the Lie derivative over $V$ up to the commutator of the two vector fields.

\begin{corollary}
Suppose that $\C{L}_V\omega = 0$ (in $\hat Q$-cohomology). Then $\C{L}_V\imath_X\omega = 0$ if vector fields $X$ and $V$ commute. In particular, this occurs if $\llbracket V,X\rrbracket\in\ff{I}$ (or $\llbracket V,X\rrbracket \sim 0$ in equivalence relation \eqref{equivalence of multivectors}).
\end{corollary}

\begin{corollary}
Suppose that $\omega$ is a projectible $1$-form such that $\C{L}_V\omega = 0$ (in $\hat Q$-cohomology). Let $X$ be another projectible vector field such that $\llbracket V,X\rrbracket \in\ff{I}$. Then $F = \imath_X\omega$ is a projectible function invariant over the flow generated by $V$. It is a physical observable representing an integral of the motion defined by $V$.
\end{corollary}

\begin{proposition}
Let $\omega$ and $\tau$ be projectible forms such that for a projectible vector field $V$, $\C{L}_V\omega = 0$ and $\C{L}_V\tau = 0$ (in $\hat Q$-cohomology). Then
    \[\C{L}_V[\omega,\tau] = 0,\]
at the level of cohomology.
\end{proposition}

\begin{proof}
Further applications of the Jacobi identity:
    \[(-1)^V\C{L}_V[\omega,\tau] = \left.\Big\{\Gamma,\big\{\{\Psi,\omega\},\tau\big\}|_{\C{L}}\Big\}\right|_{\C{L}}\]
    \[= \left.\Big\{\big\{\{\Gamma,\Psi\},\omega\big\},\tau\Big\}\right|_{\C{L}} + \left.\Big\{\big\{\Psi,\{\Gamma,\omega\}\big\},\tau\Big\}\right|_{\C{L}} + \left.\big\{\{\Psi,\omega\},\{\Gamma,\tau\}\big\}\right|_{\C{L}}.\]
The first term vanishes from equation \eqref{master equation2}, the second two from the invariance of the two forms over $V$ and the projectibility conditions.
\end{proof}

\begin{corollary}
For a projectible $1$-form $\omega$ and function $F$ such that $\dot\omega = 0$ and $\dot F = 0$ in cohomology, the function $[\omega,F]$ is projectible and $\C{L}_V[\omega,F] = 0$. That is, the Koszul bracket of a $V$-invariant form with a $V$-invariant function produces again a $V$-invariant function - a physical observable which is an integral of the motion prescribed by $V$.
\end{corollary}

\section{Examples of Weak Poisson Systems}

\subsection{Invertible Upper Triangular Matrices}
Consider the Lie group $T_n$ of invertible $n\times n$ upper triangular matrices and its Lie algebra $\ff{t}$, the space of all upper triangular matrices. Let $\ff{h} = [\ff{t},\ff{t}]$ be the first derived subalgebra of $\ff{t}$, the Lie algebra of strictly upper triangular matrices.

A gauge system can be defined on $T_n$ by specifying no constraint equations and setting the gauge generators to be the left invariant vector fields in $\ff{h}$. Take $V\in\ff{t}$ such that $V$ can be written as a sum of a non-zero diagonal matrix and a strictly upper triangular matrix. That we require a non-zero diagonal matrix is the same as asking that $V$ is not purely a gauge generator and contains non-trivial components. As an example, take $V = (0 \cdots 0\, v)$ where $v$ is a column vector $v = (v_1, \ldots,v_n)^{T}$ and $v_n\neq 0$. Consider a bivector $P = P^{ij}e_j\wedge e_i$ where $e_i$ is the zero matrix with a $1$ in the $(i,i)^{th}$ position and $P^{ij}=-P^{ji}$ is an antisymmetric constant matrix. Trivially, both of these are projectible:
    \[[V,R] \in \ff{h}, \qquad \llbracket P,R\rrbracket = P^{ij}[e_j,R]\wedge e_i - P^{ij}[e_i,R]\wedge e_j\in \ff{h}\wedge\ff{t},\qquad R\in\ff{h},\]
since every commutator belongs to the derived subalgebra $\ff{h}$.

It is also trivial that $P$ acts as a weak Poisson bivector
    \[\llbracket P,P\rrbracket = P^{ij}P^{mn}[e_i,e_m]\wedge e_j\wedge e_n + \cdots \in \ff{h}\wedge\ff{t}\wedge\ff{t},\]
and that $V$ is a weak Poisson vector field for $P$
    \[\llbracket V,P\rrbracket = 2P^{ij}[V,e_j]\wedge e_i \in \ff{h}\wedge\ff{t}.\]
Therefore we may define a weak Poisson structure on $T_n$. (This weak Poisson system contains a lot of freedom in the choice of Poisson bivector and the vector field $V$.) The ideal $\ff{h}$ defines a normal subgroup $H\subset T_n$ consisting of the upper triangular matrices with $1$'s on the diagonal  and we can consider the quotient $T_n/H$ equipped with the induced Poisson structure.

This example clearly admits a generalisation. Let $G$ be a Lie group with Lie algebra $\ff{g}$ such that $[\ff{g},\ff{g}]$ is a proper ideal of $\ff{g}$. Since $\ff{g}/[\ff{g},\ff{g}]\neq 0$, a non-trivial vector field and Poisson bivector $P$ may then be constructed in such a way that all the relations between them are trivially satisfied if we consider the gauge generators to be the vector fields in the derived subalgebra. Since the derived subalgebra is an ideal, it corresponds to a normal subgroup $H$ in $G$ and we can consider the quotient $X = G/H$ as a Poisson manifold with Poisson structure $P$. (In general $H$ may not always be a closed subgroup in the Lie group sense, therefore we assume that $H$ is a proper Lie subgroup of the Lie group $G$.) A one-parameter subgroup of Poisson automorphisms of $C^\infty(X)$ is generated by the projectible vector field.

Note that any ideal of $\ff{g}$ can be suitably chosen to be a set of gauge generators, however the choice of bivector becomes less trivial as the projectibility conditions are not automatically satisfied.

\subsection{The Heisenberg Group}
Consider the $3$-dimensional Heisenberg group $H_3(\mathbb{R})$:
    \[H_3(\mathbb{R}) = \left\{ \left.\left(\begin{array}{ccc}1&a&c\\0&1&b\\0&0&1\end{array}\right)\right|a,b,c\in\mathbb{R}\right\}.\]
Its Lie algebra $\ff{h}_3$ is generated by
    \[p = \left(\begin{array}{ccc}0&1&0\\0&0&0\\0&0&0\end{array}\right), \quad q = \left(\begin{array}{ccc}0&0&0\\0&0&1\\0&0&0\end{array}\right),\quad z = \left(\begin{array}{ccc}0&0&1\\0&0&0\\0&0&0\end{array}\right)\]
which satisfy the commutation relation
    \[[p,q]=z,\]
and all others are zero.

Identify $H_3(\mathbb{R})$ with $\mathbb{R}^3$ with local coordinates $(x,y,z)$. The generators of $\ff{h}_3$ correspond to the vector fields
    \[X = \p_x-\F{1}{2}y\p_z,\quad Y = \p_y+\F{1}{2}x\p_z, \quad Z = \p_z.\]
Consider $Z$ as a gauge generator on $\mathbb{R}^3$. Then a bivector $P$ such as
    \[P = X\wedge Y\]
is a weak Poisson bivector as it satisfies the weak Jacobi identity
    \[\llbracket P,P\rrbracket = 2X\wedge Y\wedge Z.\]
Let
    \[V = y\p_x - x\p_y, \qquad \theta = dz + \F{1}{2}\left(ydx - xdy\right)\]
be a vector field and $1$-form respectively.  Both are gauge invariant, and the contraction
    \[\theta(V) = \F{1}{2}\left(x^2+y^2\right)\]
is invariant with respect to the flow of $V$: $V(x^2+y^2) = 0$. We can see that $\mathbb{R}^2$ as the quotient of $\mathbb{R}^3$ by the gauge orbits is equipped with the Poisson structure $P = \p_x\wedge \p_y$ and a subgroup of Poisson automorphisms generated by $V$. This is the Poisson structure induced from the canonical symplectic structure on $\mathbb{R}^2$. Under the flow of $V$, functions of the form $f=f(x^2+y^2)$ are invariant which is clear since $V$ generates rotations about the origin.

Both of these examples contained no constraint equations, however these can be easily included by viewing each group as a subgroup of the corresponding general linear group and smoothly extending the vector fields into this ambient group.

\subsection{Jacobi Manifolds}
A Jacobi manifold $(M,P,R)$ is a manifold $M$ with a distinguished bivector field $P$ and vector field $R$ such that
    \[\llbracket P,P\rrbracket = 2P\wedge R,\qquad \llbracket P,R\rrbracket = 0.\]
On setting $R$ to be a single gauge generator, we see that $P$ together with $R$ defines a weak Poisson structure on the manifold $M$. Consider the case where $M$ is a contact manifold, so $\dim M = 2n+1$ and there exists a $1$-form $\theta$ such that $\theta\wedge(d\theta)^n$ does not vanish. Locally on $M$, there exist coordinates $(t,q^1,\ldots,q^n,p^1,\ldots,p^n)$ such that
    \[\theta = dt - \sum_{i=1}^np^idq^i.\]
Let
    \[R = \F{\p}{\p t},\qquad P = \sum_{i=1}^n\left(\F{\p}{\p q^i} + p^i\F{\p}{\p t}\right)\wedge\F{\p}{\p p^i}.\]
Then $P$ and $R$ as defined satisfy the conditions of a Jacobi manifold and so specify a weak Poisson structure on $M$. Further we are provided with a natural gauge invariant $1$-form $\theta$
    \[\C{L}_{\p_t}\theta = 0.\]
Any vector field of the form $V = V^i_q(q,p)\F{\p}{\p q^i} + V^j_p(q,p)\F{\p}{\p p^j} + V(q,p,t)\p_t$ will be gauge invariant, since $[V,\p_t] = \p_tV(q,p,t)\p_t$. Consider the vector field
    \[V = q^i\F{\p}{\p q^i} - p^i\F{\p}{\p p^i}.\]
This defines a weak Poisson vector field since $\llbracket V,P\rrbracket = 0$. Contracting with $\theta$ gives
    \[f = \theta(V) = -\sum_{i=1}^np^iq^i,\]
and $V(f) = 0$, since $V=\llbracket P,f\rrbracket$. So we have a one-parameter family of Poisson automorphisms on the quotient space with $f$ a constant function along the flow of $V$.

\section*{Acknowledgments}
It is a pleasure to acknowledge Th.~Voronov for many helpful
discussions.
M.P. appreciates the hospitality of Tomsk State
University where he has begun this work, and thanks D.~Kaparulin for
all the help he provided during this visit. The visit by M.P. to
Tomsk State University was supported by the RFBR grant 14-31-50799.
The work of S.L. and A.Sh. was partially supported by the RFBR grant
13-02-00551.





\bibliography{WKB}

\end{document}